\documentclass[aip,jmp,12pt,superscriptaddress]{revtex4-1}
\usepackage{amsfonts}
\usepackage{graphicx}
\usepackage{amsmath}
\usepackage{amsthm}
\usepackage{bbm}
\usepackage{hyperref}

\newtheorem{thm}{Theorem}[section]
\newtheorem{cor}{Corollary}
\newtheorem{lem}{Lemma}
\newtheorem{prop}{Proposition}
\theoremstyle{definition}

\theoremstyle{remark}
\newtheorem{rem}{Remark}

\begin{document}

\title{Normal form decomposition for Gaussian-to-Gaussian superoperators}
\author{Giacomo De Palma}
\affiliation{NEST, Scuola Normale Superiore and Istituto Nanoscienze-CNR, I-56126 Pisa, Italy}
\affiliation{INFN, Pisa, Italy}

\author{Andrea Mari}
\affiliation{NEST, Scuola Normale Superiore and Istituto Nanoscienze-CNR, I-56126 Pisa, Italy}

\author{Vittorio Giovannetti}
\affiliation{NEST, Scuola Normale Superiore and Istituto Nanoscienze-CNR, I-56126 Pisa, Italy}

\author{Alexander S. Holevo}
\affiliation{Steklov Mathematical Institute, 119991 Moscow, Russia and National Research University Higher School of Economics (HSE),
101000 Moscow, Russia}

\begin{abstract}
In this paper we explore the set of linear maps sending the set of quantum Gaussian states into itself. These maps are in general not positive, a feature which can be exploited as a test to check whether a given quantum state belongs to the convex hull of Gaussian states (if one of the considered maps sends it into a non positive operator, the above state is certified not to belong to the set). Generalizing a result known to be valid under the assumption of complete positivity, we provide a characterization of these Gaussian-to-Gaussian (not necessarily positive) superoperators in terms of their action on the characteristic function of the inputs. For the special case of one-mode mappings we also show that any Gaussian-to-Gaussian superoperator can be expressed as a concatenation of a phase-space dilatation, followed by the action of a completely positive Gaussian channel, possibly composed with a transposition. While a similar decomposition is shown to fail in the multi-mode scenario, we prove that it still holds at least under the further hypothesis of homogeneous action on the covariance matrix.
\end{abstract}

\maketitle

\section{Introduction}

Gaussian Bosonic States (GBSs) play a fundamental role in the study of
continuous-variable (CV) quantum information processing~\cite
{HOLEVOb1,PARISBOOK,REV,REV1} with applications in quantum cryptography,
quantum computation and quantum communication where they are known to
provide optimal ensembles for a large class of quantum communication lines
(specifically the phase-invariant Gaussian Bosonic maps)~\cite
{GAUSC,GAUS0,GAUS00,WOLF1,GAUS1,GAUS2}. GBSs are characterized by the
property of having Gaussian Wigner quasi-distribution and describe Gibbs
states of Hamiltonians which are quadratic in the field operators of the
system. Further, in quantum optics they include coherent, thermal and
squeezed states of light and can be easily created via linear amplification
and loss.

Directly related to the definition of GBSs is the notion of Gaussian
transformations~\cite{HOLEVOb1,REV,REV1}, i.e. superoperators mapping the
set $\mathfrak{G}$ of GBSs into itself. In the last two decades, a great
deal of attention has been devoted to characterizing these objects. In
particular the community focused on Gaussian Bosonic Channels (GBCs)~\cite
{GAUSC}, i.e. Gaussian transformations which are completely positive (CP)
and provide hence the proper mathematical representation of data-processing
and quantum communication procedures which are physically implemetable~\cite
{CAVES}. On the contrary, less attention has been devoted to the study of
Gaussian superoperators which are not CP or even non-positive. A typical
example of such mappings is provided by the phase-space dilatation, which,
given the Wigner quasi-distribution $W_{\hat{\rho}}(\mathbf{r})$ of a
state $\hat{\rho}$ of $n$ Bosonic modes, yields the function $
W^{(\lambda)}_{\hat{\rho}}( \mathbf{r}) \equiv W_{\hat{\rho}}(\mathbf{r}
/\lambda)/\lambda^{2n}$ as an output, with the real parameter $\lambda$
satisfying the condition $|\lambda|>1$.
On one hand, when acting on $\mathfrak{G}$ the mapping
\begin{eqnarray}
W_{\hat{\rho}}(\mathbf{r})\mapsto W^{(\lambda)}_{\hat{\rho}}( \mathbf{r}) \;,
\label{MAP}
\end{eqnarray}
always outputs proper (Gaussian) states. Specifically, given $\hat{\rho}\in
\mathfrak{G}$ one can identify another Gaussian density operator $\hat{\rho}
^{\prime}$ which admits the function $W^{(\lambda)}_{\hat{\rho}}( \mathbf{r}
) $ as Wigner distribution, i.e. $W_{\hat{\rho}^{\prime}}(\mathbf{r})=
W^{(\lambda)}_{\hat{\rho}}( \mathbf{r})$. On the other hand, there exist
inputs $\hat{\rho}$ for which $W^{(\lambda)}_{\hat{\rho}}( \mathbf{r})$ is
no longer interpretable as the Wigner quasi-distribution of \textit{any}
quantum state: in this case in fact $W^{(\lambda)}_{\hat{\rho}}( \mathbf{r})$
is the Wigner quasi-distribution $W_{\hat{\theta}}(\mathbf{r})$
of an operator $\hat{\theta}$ which is not positive \cite{WERNER} (for example, any pure non-Gaussian state has this property for any $\lambda\neq\pm1$ \cite{dias}).
Accordingly
phase-space dilatations~(\ref{MAP}) should be considered as ``unphysical''
transformations, i.e. mappings which do not admit implementations in the
laboratory. Still dilatations and similar exotic Gaussian-to-Gaussian
mappings turn out to be useful mathematical tools that can be employed to
characterize the set of states of CV systems in a way which is not
dissimilar to what happens for positive (but not completely positive)
transformations in the analysis of entanglement~\cite{HORO}. In particular
Br\"ocker and Werner~\cite{WERNER} used~(\ref{MAP}) to study the convex hull
$\mathfrak{C}$ of Gaussian states (i.e. the set of density operators $\hat{
\rho}$ which can be expressed as a convex combination of elements of $
\mathfrak{G}$). The rationale of this analysis is that the set $\mathfrak{F}$
of density operators which are mapped into proper output states by this
transformation includes $\mathfrak{C}$ as a proper subset, see Fig.~\ref
{FigSys}. Accordingly if a certain input $\hat{\rho}$ yields a $
W^{(\lambda)}_{\hat{\rho}}( \mathbf{r})$ which is not the Wigner
distribution of a state, we can conclude that $\hat{\rho}$ is not an element
of $\mathfrak{C}$.
Finding mathematical and experimental criteria which help
in identifying the boundaries of $\mathfrak{C}$ is indeed a timely and
important issue which is ultimately related with the characterization of
non-classical behavior in CV systems, see e.g. Refs.~\cite
{RADIM1,DET1,DET2,DET3,mari,kiesel,richter,jezek,NONGAUS1,NONGAUS2,NONGAUS3}, and also \cite{ferm1,ferm2} for the fermionic case.

In this context a classification of non-positive Gaussian-to-Gaussian
operations is mandatory.
This analysis has been initiated in \cite{cirac}, where Gaussian-to-Gaussian maps are characterized through their
Choi-Jamio\l kowski state, under the hypothesis that this state has a Gaussian characteristic function.
One goal of this paper is proving this hypothesis: we prove
that the action of such
transformations on the covariance matrix and on the first moment must be
linear, and we write explicitly the transformation properties of the
characteristic function (Theorem \ref{gaussthmq}).
In the classical case, any probability measure can be written as a convex
superposition of Dirac deltas, so the convex hull of the Gaussian measures
coincides with the whole set of measures. A simple consequence of this fact is that a linear transformation
sending Gaussian measures into Gaussian (and then positive) measures is
always positive. Nothing of this holds in the more interesting quantum case,
so we focus on it, and use Theorem \ref{gaussthmq} to get a decomposition
which, for single-mode operations, shows that any linear quantum
Gaussian-to-Gaussian transformation can always be
decomposed as a proper combination of a dilatation~(\ref{MAP}) followed by a
CP Gaussian mapping plus possibly a transposition. We also show that our
decomposition theorem applies to the multi-mode case, as long as we restrict
the analysis to Gaussian transformations which are homogeneous at the level
of covariance matrix. For completeness we finally discuss the case of
contractions: these are mappings of the form~(\ref{MAP}) with $|\lambda|<1$.
They are not proper Gaussian transformations because they map some Gaussian
states into non-positive operators. Still some of the results which apply to
the dilatations can be extended to this set.

\begin{figure}[t]
\begin{center}
\includegraphics[trim=0pt 0pt 0pt 0pt, clip, width=0.7\textwidth]{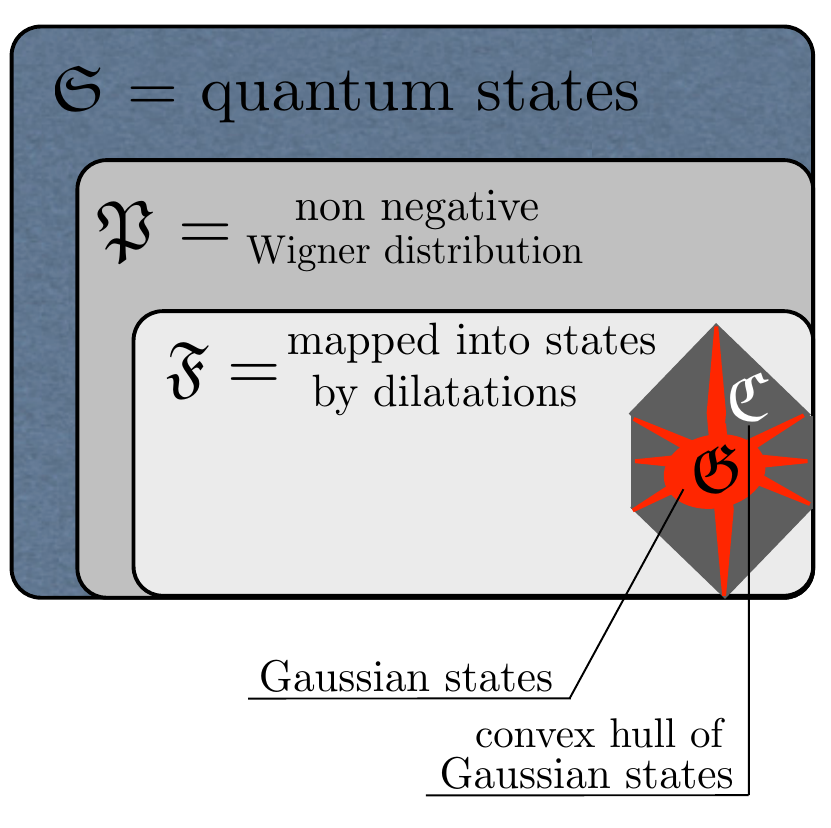}
\end{center}
\caption{(Color online) Pictorial representation of the structure of the set
of states $\mathfrak{S}$ of a CV system. $\mathfrak{P}$ is the subset of
density operators $\hat{\protect\rho}$ which have non-negative Wigner
distribution~(\protect\ref{WIGWIG}). $\mathfrak{F}$ is set of states which instead are mapped into proper density operators by an arbitrary
dilatation~(\protect\ref{MAP}).
$\mathfrak{G}$ is the set of Gaussian states
and $\mathfrak{C}$ its convex hull.
$\mathfrak{S}$, $\mathfrak{P}$, $
\mathfrak{F}$, and $\mathfrak{C}$ are closed under convex convolution, $
\mathfrak{G}$ is not. For a detailed study of the relations among these sets
see Ref.~\protect\cite{WERNER}. }
\label{FigSys}
\end{figure}

The manuscript is organized as follows. We start in Sec.~\ref{S:NOTATION} by
introducing the notation and recalling some basic facts about characteristic functions, Wigner
distributions and Gaussian states. In Sec.~\ref{S:PROBLEM}
we state the problem and prove the Theorem  \ref{gaussthmq} characterizing the action of Gaussian-to-Gaussian
superoperators on the characteristic functions of quantum states and its variations, including the probabilistic analog. In subsection \ref{contr} we consider the case of contractions. In Sec.~\ref{S:ONE} we present the main result of the
manuscript, i.e. the decomposition theorem for single-mode
Gaussian-to-Gaussian transformations. The multi-mode case is then analyzed
in Sec.~\ref{S:MULTI}. The paper ends hence with Sec.~\ref{S:CON} where we
present a brief summary and discuss some possible future developments.
In Appendix \ref{app} we prove the unboundedness of phase-space dilatations with respect to the trace norm.

\section{Notation}

\label{S:NOTATION} In this section we introduce the notation and some basic
definitions which are useful to set the problem.

\subsection{Symplectic structure}

 Consider a CV system constituted by $n$ independent modes~\cite
{HOLEVOb1,PARISBOOK} with quadratures satisfying the canonical commutation
relations
\begin{equation}
\left[ \hat{Q}^{j},\;\hat{P}^{k}\right] =i\delta ^{jk}\;\hat{\mathbbm{1}}
\;,\qquad j,\,k=1,\ldots ,n\;.  \label{[QP]}
\end{equation}
Organizing these operators in the column vector $\mathbf{\hat{R}}$ of
elements
\begin{equation}
\mathbf{\hat{R}}\equiv \left( \hat{Q}^{1},\,\hat{P}^{1},\ldots ,\hat{Q}
^{n},\,\hat{P}^{n}\right) ^{T}\;,
\end{equation}
equation~(\ref{[QP]}) can be cast in the equivalent form
\begin{equation}
\left[\hat{R}^i,\;\hat{R}^j\right] =i\Delta^{ij} \;\hat{
\mathbbm{1}}\;,  \label{[RR]}
\end{equation}
where $\Delta $ is the skew-symmetric matrix
\begin{equation}
\Delta =\bigoplus_{i=1}^{n}\left(
\begin{array}{cc}
0 & 1 \\
-1 & 0
\end{array}
\right) \;,  \label{deltacan}
\end{equation}
which defines the symplectic form of the problem. Accordingly a real matrix $
S$ is said to be symplectic if it preserves $\Delta $, i.e. if the following
identity holds
\begin{equation}
S\Delta S^{T}=\Delta \;.  \label{symplectic}
\end{equation}

\subsection{Characteristic and Wigner functions}

For any density operator $\hat{\rho}\in \mathfrak{S}$ we define its
symmetrically ordered characteristic function as
\begin{equation}
\chi (\mathbf{k})\equiv \mathrm{Tr}\left(\hat{\rho}\;\exp \left( i\,\mathbf{k}
^{T}\,\mathbf{\hat{R}}\right)\right)\;,\qquad\mathbf{k}\in \mathbbm{R}^{2n}\;.  \label{char}
\end{equation}
This formula makes sense for any trace-class operator $\hat{\rho},$ the
associated $\chi (\mathbf{k})$ being a bounded continuous function. By using
the Parceval-type formula \cite{HOLEVOb1}
\begin{equation}
\mathrm{Tr}\left(\hat{\rho}_{1}^{\dag}\;\hat{\rho}_{2}\right)=\int \overline{\chi _{
\hat{\rho}_{1}}(\mathbf{k})}\,\chi _{\hat{\rho}_{2}}(\mathbf{k})\;\;\frac{d
\mathbf{k}}{(2\pi )^{n}}\;,  \label{HS}
\end{equation}
one can extend the correspondence (\ref{char}) to Hilbert-Schmidt operators,
the associated $\chi (\mathbf{k})$ being a square-integrable function. The
correspondence is the isomorphism between the Hilbert space $\mathfrak{H}$
of Hilbert-Schmidt operators with the inner product given by the left hand side of
(\ref{HS}), and that of square-integrable functions of $\mathbf{k}$. The operator
$\hat{\rho}$ can be expressed as
\begin{equation}
\hat{\rho}=\int \chi _{\hat{\rho}}(\mathbf{k})\;e^{-i\mathbf{k}^{T}\mathbf{
\hat{R}}}\;\frac{d\mathbf{k}}{(2\pi )^{n}}\;.
\end{equation}
We also define its Wigner function as the Fourier transform of the characteristic function, square integrable for any $\hat{\rho}\in\mathfrak{H}$:
\begin{equation}
W_{\hat{\rho}}(\mathbf{r}) = \int \chi _{\hat{\rho}}(\mathbf{k})\;e^{-i
\mathbf{k}^{T}\mathbf{r}}\;\frac{d\mathbf{k}}{(2\pi )^{2n}}\;.  \label{WIGWIG}
\end{equation}

\subsection{Heisenberg uncertainty}

The covariance matrix of $\hat{\rho}$ is defined as
\begin{equation}
\sigma^{ij} \equiv \mathrm{Tr}\left( \hat{\rho}\;\left\{ \hat{R}^i
-\left\langle \hat{R}^i\right\rangle ,\;\left( \hat{R}^j
-\left\langle \hat{R}^j\right\rangle \right)\right\} \right) \;,
\label{covm}
\end{equation}
with $\left\langle \mathbf{\hat{R}}\right\rangle \equiv \mathrm{Tr}\left[
\hat{\rho}\;\mathbf{\hat{R}}\right] ,$ provided the second moments of $\hat{\rho}
$ are finite. The positivity of $\hat{\rho}$ together with the commutation
relations \eqref{[RR]} imply the Robertson-Heisenberg uncertainty relation
\begin{equation}
\sigma \geq \pm i\Delta \;,  \label{heis}
\end{equation}
where the inequality is meant to hold for both plus and minus signs in the right-hand-side.
We also recall that the Williamson theorem~\cite{WILL} ensures that given a
symmetric positive definite matrix $\sigma $, there exists a symplectic
matrix $S$ such that
\begin{equation}
S\sigma S^{T}=\bigoplus_{i=1}^{n}\nu _{i}\mathbbm{1}_{2}\;,\qquad \nu
_{i}>0\;.
\end{equation}
The $\nu _{i}$ are called symplectic eigenvalues of $\sigma $ and can be
computed as the positive ordinary eigenvalues of the matrix $i\sigma\Delta^{-1}$
(they come in couples of opposite sign). The Heisenberg uncertainty
principle \eqref{heis} is hence satisfied iff they are all greater than one:
\begin{equation}
\nu _{j}\geq 1\;,\qquad j=1,\ldots ,n\;.  \label{seig}
\end{equation}
The symplectic condition \eqref{symplectic} simplifies in the case of one
mode. Indeed for any $2\times 2$-matrix $M$ we have
\begin{equation}
M\Delta M^{T}=\Delta \det M\;,  \label{det1mode}
\end{equation}
therefore a $2\times 2$-matrix $S$ is symplectic iff
\begin{equation}
\det S=1\;.
\end{equation}
We recall also that a symmetric positive definite $2\times 2$-matrix $\sigma
$ has a single symplectic eigenvalue, given by
\begin{equation}
\nu ^{2}=\det \sigma \;;
\end{equation}
then from \eqref{seig} $\sigma \geq 0$ is the covariance matrix of a quantum
state iff
\begin{equation}
\det \sigma \geq 1\;.  \label{state1}
\end{equation}

\subsection{Gaussian states}

States with positive Wigner function~(\ref{WIGWIG}) form a convex subset $
\mathfrak{P}$ in the space of the density operators $\mathfrak{S}$ of the
system. The set $\mathfrak{G}$ of Gaussian states is a proper subset of $
\mathfrak{P}$. A Gaussian state $\hat{\rho}_{\mathbf{x},\,\sigma }\in
\mathfrak{G}$ with covariance matrix $\sigma $ and first moments $\mathbf{x}$
is defined by the property of having Gaussian characteristic function
\begin{equation}
\chi (\mathbf{k})=e^{-\frac{1}{4}\mathbf{k}^{T}\sigma \mathbf{k}+i\mathbf{k}
^{T}\mathbf{x}}\;,
\end{equation}
i.e. Gaussian Wigner function
\begin{equation}
W(\mathbf{r})=\frac{1}{\sqrt{\det\left(\pi\,\sigma\right)}}\;e^{-(\mathbf{r}-\mathbf{x}
)^{T}\sigma ^{-1}(\mathbf{r}-\mathbf{x})}\;.
\end{equation}

For $\sigma = \mathbbm{1}_{2n}$ we obtain the family of coherent states $\hat{\rho}_{\mathbf{x},\,\mathbbm{1}_{2n}},\;\mathbf{x}\in \mathbbm{R}^{2n}$.
The Husimi function $\mathrm{Tr}\left(\hat{\rho}_{\mathbf{x},\,\mathbbm{1}_{2n} }\hat{\rho}\right)$ of any bounded operator $\hat{\rho}$ uniquely defines $\hat{\rho}$ .
It follows that the linear span of the set of coherent states, and hence of all Gaussian states, is dense in the Hilbert space of
Hilbert-Schmidt operators $\mathfrak{H}$. Similarly, these linear spans are dense in the Banach space of trace-class
operators $\mathfrak{T}$.

\subsection{The Wigner-positive states and the convex hull of Gaussian states
}

\label{subsec1} Starting from the vacuum, devices as simple as
beam-splitters combined with one-mode squeezers permit (at least in
principle) to realize all the elements of $\mathfrak{G}$. Then, choosing
randomly according to a certain probability distribution which Gaussian
state to produce, it is in principle possible to realize all the states in
the convex hull $\mathfrak{C}$ of the Gaussian ones, i.e. all the states $
\hat{\rho}$ that can be written as
\begin{equation}
\hat{\rho}=\int\hat{\rho}_{\mathbf{x},\,\sigma}\;d\mu(\mathbf{x},\sigma)\;,
\label{gaussenv}
\end{equation}
where $\mu$ is the associated probability measure of the process.

It is easy to verify that $\mathfrak{C}$ is strictly larger than $\mathfrak{G
}$, i.e. there exist states (\ref{gaussenv}) which are not Gaussian. On the
other hand, one can observe that \eqref{gaussenv} implies
\begin{equation}
W_{\hat{\rho}}(\mathbf{r})=\int \frac{1}{\sqrt{\det\left(\pi\,\sigma\right)}}
\;e^{-(\mathbf{r}-\mathbf{x})^T\sigma^{-1}(\mathbf{r}-\mathbf{x})}\;d\mu(
\mathbf{x},\sigma)>0\;,
\end{equation}
so also $\mathfrak{C}$ is included into $\mathfrak{P}$, see Fig.~\ref{FigSys}
. There are however elements of $\mathfrak{P}$ which are not contained in $
\mathfrak{C}$: for example, any finite mixture of Fock states
\begin{equation}
\hat{\rho}=\sum_{n=0}^Np_n|n\rangle\langle n|\qquad N<\infty\qquad
p_n\geq0\qquad\sum_{n=0}^Np_n=1
\end{equation}
is not even contained in the weak closure of $\mathfrak{C}$, even if some of
them have positive Wigner function~\cite{WERNER}.

\section{Characterization of Gaussian-to-Gaussian maps}

\label{S:PROBLEM}

Determining whether a given state $\hat{\rho}$ belongs to the convex hull $
\mathfrak{C}$ of the Gaussian set is a difficult problem~\cite
{RADIM1,DET1,DET2,DET3}. Then, there comes the need to find criteria to
certify that $\hat{\rho}$ \emph{cannot} be written in the form
\eqref{gaussenv}. A possible idea is to consider a non-positive
superoperator $\Phi$ sending any Gaussian state into a state~\cite{WERNER}.
By linearity $\Phi$ will also send any state of $\mathfrak{C}$ into a state,
therefore if $\Phi(\hat{\rho})$ is not a state, $\hat{\rho}$ cannot be an
element of $\mathfrak{C}$: in other words, the transformation $\Phi$ acts as
a mathematical \textit{probe} for $\mathfrak{C}$. In what follows we will
focus on those probes which are also Gaussian transformations, i.e. which
not only send $\mathfrak{G}$ into states, but which ensure that the output
states $\Phi(\hat{\rho})$ are again elements of $\mathfrak{G}$. Then the
following characterization theorem holds

\begin{thm}
\label{gaussthmq} Let $\Phi $ be a linear bounded map of
the space $\mathfrak{H}$ of Hilbert-Schmidt operators, sending the set of Gaussian states $\mathfrak{G}$
into itself. Then its action in terms of the characteristic function~(\ref{char}), the first moments and
the covariance matrix \eqref{covm} is of the form
\begin{eqnarray}
\Phi &:&\chi (\mathbf{k})\rightarrow \chi \left( K^{T}\mathbf{k}\right)
\;e^{-\frac{1}{4}\mathbf{k}^{T}\alpha \mathbf{k}+i\mathbf{k}^{T}\mathbf{y}_0
}\;,  \label{channel} \\
\Phi &:&\mathbf{x} \rightarrow K\mathbf{x}+\mathbf{y}_0 \label{Phix}\\
\Phi &:&\sigma \rightarrow K\sigma K^{T}+\alpha \;,  \label{Phicm}
\end{eqnarray}
where $\mathbf{y}_0$ is an $\mathbb{R}^{n}$ vector, and $K$ and $\alpha $ are $
2n\times 2n$ real matrices such that $\alpha $ is symmetric, and for any $
\sigma \geq \pm i\Delta $
\begin{equation}
K\sigma K^{T}+\alpha \geq \pm i\Delta \;,  \label{positivity}
\end{equation}
where the inequalities are meant to hold for both plus and minus signs in the right-hand-sides.
\end{thm}

The condition \eqref{positivity} imposes that $\Phi (\hat{\rho})$ is a
Gaussian state for any Gaussian $\hat{\rho}$.
It is weaker than the condition which guarantees complete positivity (see Eq.~\eqref{CPTP} below), which also ensures the mapping of Gaussian states into Gaussian states. An example of not completely positive mapping fulfilling \eqref{positivity} is provided by the dilatations defined in Eq.~\eqref{MAP}.
Such mappings in fact,
while explicitly not CP~\cite{WERNER}, correspond to the transformations~\eqref{channel} where we set $\mathbf{y}_0=\mathbf{0}$ and take
\begin{equation}
K=\lambda \mathbbm{1}_{2n}\;,\qquad \alpha =0\;,  \label{DILATATION}
\end{equation}
with $|\lambda |>1$. At the level of the covariance matrices~(\ref{Phicm}),
this implies $\sigma ^{\prime }=\lambda ^{2}\sigma $ which clearly still
preserve the Heisenberg inequality (\ref{heis}) (indeed ${\lambda }
^{2}\sigma \geq \sigma \geq \pm i\Delta $), ensuring hence the condition~\eqref{positivity}. Dilatations are not bounded with respect to the trace norm (see Theorem \ref{unbounded} of Appendix \ref{app}). This explains why Theorem \ref{gaussthmq} is formulated on the space of Hilbert-Schmidt operators. Indeed, via the Parceval formula we can prove that dilatations are bounded in this space:
\begin{eqnarray*}
\left\Vert \Phi (\hat{\rho})\right\Vert ^{2} &=& \int \,\left\vert \chi _{\hat{
\rho}}(\lambda \mathbf{k})\right\vert^{2}\;\frac{d\mathbf{k}}{(2\pi )^{n}}
=\\
&=&\int \,\left\vert \chi _{\hat{\rho}}(\mathbf{k})\right\vert^{2}\;\frac{d
\mathbf{k}}{(2\pi \lambda ^{2})^{n}}=\frac{1}{\lambda ^{2n}}\left\Vert \hat{
\rho}\right\Vert ^{2}\;.
\end{eqnarray*}
For $\lambda=\frac{1}{\mu}$ with $|\mu|>1$ the transformation (\ref
{DILATATION}) yields a contraction of the output Wigner quasi-distribution.
In the Hilbert space $\mathfrak{H}$, the contraction by $\lambda$ is $\lambda^{2n}$ times the adjoint of the dilatation
by $\mu=\frac{1}{\lambda}$, as follows from the Parceval formula (\ref{HS}). As different from the dilatations, these mappings no longer ensure that all
Gaussian states will be transformed into proper density operators. For
instance, the vacuum state is mapped into a non-positive operator (this shows in particular that the contractions and hence
the adjoint dilatations are non-positive maps).

Another example of transformation not fulfilling the CP requirement~(\ref
{CPTP}) but respecting~(\ref{positivity}) is the (complete) transposition
\begin{equation}
K=T_{2n}\qquad \alpha =0\;,  \label{TRANS}
\end{equation}
that is well-known not to be CP. Unfortunately, being positive it cannot be
used to certify that a given state is not contained in the convex hull $
\mathfrak{C}$ of the Gaussian ones. Is there anything else? We will prove
that for one mode, any channel satisfying \eqref{positivity} can be written
as a dilatation composed with a completely positive channel, possibly
composed with the transposition~(\ref{TRANS}), see Fig.~\ref{FigSys1}. We
will also show that in the multi-mode case this simple decomposition does
not hold in general; however, it still holds if we restrict to the channels
that do not add noise, i.e. with $\alpha =0$.

\begin{proof} Let the Gaussian state $\hat{\rho}_{\mathbf{x},\,\sigma }$ be
sent into the Gaussian state $\hat{\rho}_{\mathbf{y},\,\tau }$ with
covariance matrix $\tau (\mathbf{x},\,\sigma )$ and first moment $\mathbf{y}(
\mathbf{x},\,\sigma )$, with the characteristic function
\begin{equation}
\chi _{\Phi \mathbf{(\hat{\rho}_{\mathbf{x},\,\sigma })}}(\mathbf{k})\equiv
\chi _{y,\tau }(\mathbf{k})=e^{-\frac{1}{4}\mathbf{k}^{T}\,\tau \,\mathbf{k}
+i\mathbf{k}^{T}\,\mathbf{y}}\;.  \label{Phichi}
\end{equation}
We first remark that the functions $\tau (\mathbf{x},\,\sigma )$ and $
\mathbf{y}(\mathbf{x},\,\sigma )$ are continuous. The map $\Phi $ is bounded
and hence continuous in the Hilbert-Schmidt norm. The required continuity
follows from

\begin{lem} The bijection $(\mathbf{x},\,\sigma )\rightarrow \hat{\rho}_{\mathbf{x
},\,\sigma }$ is bicontinuous in the Hilbert-Schmidt norm.\end{lem}

The proof of the lemma follows from the Parceval formula by direct computation of the
Gaussian integral
\begin{equation*}
\int \,\left\vert \chi _{\hat{\rho}_{\mathbf{x},\,\sigma }}(\mathbf{k})-\chi
_{\hat{\rho}_{\mathbf{x}^{\prime },\,\sigma ^{\prime }}}(\mathbf{k}
)\right\vert^{2}\;\frac{d\mathbf{k}}{(2\pi )^{n}}\;.
\end{equation*}
Next, we have the identity
\begin{equation}\label{iden}
\int \hat{\rho}_{\mathbf{x}^{\prime },\,\sigma ^{\prime }}\,\mu _{\mathbf{x}
,\,\sigma }(d\mathbf{x}^{\prime })=\hat{\rho}_{\mathbf{x},\,\sigma ^{\prime }+\sigma
}\;,
\end{equation}
where $\mu _{\mathbf{x},\,\sigma }$ is Gaussian probability measure with the
first moments $\mathbf{x}$ and covariance matrix $\sigma ,$ which is
verified by comparing the quantum characteristic functions of both sides.

Applying to both sides of this identity the continuous map $\Phi $ we obtain
\begin{equation*}
\int \hat{\rho}_{\mathbf{y}(\mathbf{x}^{\prime },\,\sigma ^{\prime }),\,\tau
(\mathbf{x}^{\prime },\,\sigma ^{\prime })}\,\mu _{\mathbf{x},\,\sigma
}(d\mathbf{x}^{\prime })=\hat{\rho}_{\mathbf{y}(\mathbf{x},\,\sigma ^{\prime }+\sigma
),\,\tau (\mathbf{x},\,\sigma ^{\prime }+\sigma )}\;.
\end{equation*}
By taking the quantum characteristic functions of both sides, we obtain
\begin{eqnarray}
\int \chi _{\mathbf{y}(\mathbf{x}^{\prime },\,\sigma ^{\prime }),\,\tau (
\mathbf{x}^{\prime },\,\sigma ^{\prime })}(\mathbf{k})\;\mu _{\mathbf{x}
,\,\sigma }(d\mathbf{x}^{\prime })=&&\nonumber\\
=\chi _{\mathbf{y}(\mathbf{x},\,\sigma ^{\prime
}+\sigma ),\,\tau (\mathbf{x},\,\sigma ^{\prime }+\sigma )}(\mathbf{k}
)\;,&&\quad \mathbf{k\in }\mathbb{R}^{n}\;.  \label{cauchy}
\end{eqnarray}
Now notice that $\mu _{\mathbf{x},\,\sigma }$ is the
fundamental solution of the diffusion equation:
\begin{equation}
du=\frac{1}{4}\partial _{i}d\sigma ^{ij}\partial _{j}u\;,  \label{dphi}
\end{equation}
where $d$ is the differential with respect to $\sigma $, i.e.
\begin{equation}
d=\sum_{i,\,j=1}^{m}d\sigma ^{ij}\frac{\partial }{\partial \sigma ^{ij}}\;
\end{equation}
and
\begin{equation}
\partial _{i}=\frac{\partial }{\partial x^{i}}\;,
\end{equation}
with the sum over the repeated indices. Relation (\ref{cauchy}) means
that for any fixed $\mathbf{k},$ the function $u(\mathbf{x},\,\sigma )=\chi
_{\mathbf{y}(\mathbf{x},\,\sigma ^{\prime }+\sigma ),\,\tau (\mathbf{x}
,\,\sigma ^{\prime }+\sigma )}(\mathbf{k})$ is the solution of the Cauchy
problem for the equation (\ref{cauchy}) with the initial condition $u(
\mathbf{x},\,0)=\chi _{\mathbf{y}(\mathbf{x},\,\sigma ^{\prime }),\,\tau (
\mathbf{x},\,\sigma ^{\prime })}(\mathbf{k}).$ Since the last function is
bounded and continuous, the solution of the Cauchy problem is infinitely
differentiable in $(\mathbf{x},\,\sigma)$ for $\sigma >0.$ Substituting
\begin{equation*}
u(\mathbf{x},\,\sigma )=\exp \left[ -\frac{1}{4}\mathbf{k}^{T}\,\tau (
\mathbf{x},\,\sigma ^{\prime }+\sigma )\,\mathbf{k}+i\mathbf{k}^{T}\,\mathbf{
y}(\mathbf{x},\,\sigma ^{\prime }+\sigma )\right]
\end{equation*}
into (\ref{dphi}) and differentiating the exponent, we obtain the identity
\begin{equation*}
-\frac{1}{4}\mathbf{k}^{T}\;d\tau \;\mathbf{k}+i\mathbf{k}^{T}d\mathbf{y} =
\end{equation*}\begin{eqnarray*}
&=&\frac{1}{4}\left(\frac{1}{4}\mathbf{k}^T\partial_i\tau\,\mathbf{k}-i
\mathbf{k}^T\partial_i
\mathbf{y}\right)d\sigma^{ij}\left(\frac{1}{4}\mathbf{k}^T\partial_j\tau\,
\mathbf{k}-i\mathbf{k}^T\partial_j\mathbf{y}\right)\\
&&-\frac{1}{16}\mathbf{k}^T\left(\partial_id\sigma^{ij}\partial_j\tau\right)
\mathbf{k}+
\frac{i}{4}\mathbf{k}^T\partial_id\sigma^{ij}\partial_j\mathbf{y}\;.
\end{eqnarray*}
We can now compare the two expressions. Since the left hand side
contains only terms at most quadratic in $\mathbf{k}$, we get
\begin{equation}
\partial _{i}\tau =0\;,
\end{equation}
i.e. $\tau $ does not depend on $\mathbf{x}$. Then, the right hand side
simplifies into
\begin{equation}
-\frac{1}{4}\mathbf{k}^{T}\left( \partial _{i}\mathbf{y}\,d\sigma
^{ij}\partial _{j}\mathbf{y}^{T}\right) \mathbf{k}+\frac{i}{4}\mathbf{k}
^{T}\partial _{i}d\sigma ^{ij}\partial _{j}\mathbf{y}\;.  \label{logphi3}
\end{equation}
Comparing again with the left hand side, we get
\begin{eqnarray}
d\tau (\sigma ) &=&\partial _{i}\mathbf{y}\,d\sigma ^{ij}\partial _{j}
\mathbf{y}^{T}  \label{deltatau} \\
d\mathbf{y}(\mathbf{x},\,\sigma ) &=&\frac{1}{4}\partial _{i}d\sigma
^{ij}\partial _{j}\mathbf{y}\;.  \label{deltay}
\end{eqnarray}
Since $d\tau (\sigma )$ does not depend on $\mathbf{x}$, also $\partial _{i}
\mathbf{y}$ cannot, i.e. $\mathbf{y}$ is a linear function of $\mathbf{x}$:
\begin{equation}
\mathbf{y}(\mathbf{x},\,\sigma )=K(\sigma )\,\mathbf{x}+\mathbf{y}
_{0}(\sigma )\;,
\end{equation}
where $K(\sigma )$ and $\mathbf{y}_{0}(\sigma )$ are still arbitrary
functions. But now \eqref{deltay} becomes
\begin{equation}
d\mathbf{y}(\mathbf{x},\,\sigma )=0\;,
\end{equation}
i.e. $\mathbf{y}$ does not depend on $\sigma $, i.e.
\begin{equation}
\mathbf{y}=K\mathbf{x}+\mathbf{y}_{0}\;,  \label{y}
\end{equation}
with $K$ and $\mathbf{y}_{0}$ constant. Finally, \eqref{deltatau} becomes
\begin{equation}
d\tau (\sigma )=K\,d\sigma \,K^{T}\;,
\end{equation}
that can be integrated into
\begin{equation}
\tau (\sigma )=K\,\sigma \,K^{T}+\alpha \;.  \label{tau}
\end{equation}

Thus we get that the transformation rules for the first and second
moments are given by Eqs. \eqref{Phix} and \eqref{Phicm}. The positivity condition
for quantum Gaussian states implies \eqref{positivity}. The map defined by
\eqref{channel} correctly reproduces \eqref{Phix} and \eqref{Phicm}, so it
coincides with $\Phi $ on the Gaussian states. Since it is linear and
continuous, and the linear span of of Gaussian states is dense in $\mathfrak{
H},$ it coincides with $\Phi $ on the whole  $\mathfrak{H}$.
\end{proof}

\begin{rem} A similar argument can be used to prove that any linear positive map $\Phi$ of the Banach space $\mathfrak{T}$
of trace-class operators, leaving the set of Gaussian states globally invariant, has the form~(\ref{channel}).
By Lemma 2.2.1 of \cite{davies} any such map is bounded, and the proof of Theorem \ref{gaussthmq}
can be repeated, with $\mathfrak{H}$ replaced by $\mathfrak{T}$. In addition, since the trace of operator is continuous
on $\mathfrak{T}$, the formula~(\ref{channel})  implies preservation of trace. However, the positivity condition is difficult to express
in terms of the map parameters $\mathbf{y}_{0}, K, \alpha$.

On the other hand, if $\Phi$ is completely positive then the necessary and sufficient condition is
\begin{equation}
\alpha \geq \pm i(\Delta -\Delta _{K})\;,  \label{CPTP}
\end{equation}
where
\begin{equation}
\Delta _{K}\equiv K\Delta K^{T}\;.  \label{DELTAK}
\end{equation}
Thus $\Phi$ is a quantum Gaussian channel \cite{GAUSC},
and the condition Eq.~(\ref{positivity}) is replaced by the more stringent constraint \eqref{CPTP}.

For automorphisms of the $C^*$-algebra of the Canonical Commutation Relations a similar characterization,
based on a different proof using partial ordering of Gaussian states, was first given in \cite{demoen,fannes}.
\end{rem}

\begin{rem}
There is a counterpart of Theorem \ref{gaussthmq} in probability theory:
\end{rem}
\begin{thm}
Let $\Phi $ be an endomorphism (linear bounded transformation) of the Banach space $
\mathcal{M}(\mathbb{R}^{n})$ of finite signed Borel measures on $\mathbbm{R}
^{n}$ (equipped with the total variation norm) having the Feller property (the dual  $\Phi^* $ leaves invariant the space of bounded continuous functions on $\mathbbm{R}^{n}$).
Then, if $\Phi$ sends the set of Gaussian probability measures into itself,  $\Phi $ is a Markov operator whose action in terms of characteristic functions is of the form \eqref{channel}, with the condition \eqref{positivity} replaced by $\alpha\geq 0$.
\begin{proof}
The proof is parallel to the proof of Theorem \ref{gaussthmq}, with replacement of \eqref{iden} by the corresponding identity for Gaussian probability measures.
As a result, we obtain that the action of $\Phi$ in terms of characteristic functions is given by (\ref{channel}) for any measure $\mu$
which is a linear combination of Gaussian probability measures. For arbitrary measure $\mu\in \mathcal{M}(\mathbb{R}^{n})$ the characteristic function of $\Phi(\mu)$ is
\begin{equation*}
\chi_{\Phi(\mu)}(\mathbf{k})=\int \mathbf{e}^{i\,\mathbf{k}^T\mathbf{x}}\;\Phi(\mu)(d\mathbf{x}) = \int \Phi^*\left(\mathbf{e}^{i\,\mathbf{k}^T\mathbf{x}}\right)\;\mu(d\mathbf{x})\;,
\end{equation*}
where $\Phi^*\left(\mathbf{e}^{i\,\mathbf{k}^T\mathbf{x}}\right)$ is continuous bounded function by the Feller property. Since the linear span of Gaussian probability measures is dense in $\mathcal{M}(\mathbb{R}^{n})$ in the weak topology defined by continuous bounded functions (it suffices to take Dirac's deltas, i.e, probability measures degenerated at the points of  $\mathbb{R}^{n}$) , the formula (\ref{channel}) extends to characteristic function of arbitrary finite signed Borel measure on $\mathbb{R}^{n}$. The action of $\Phi$ on the moments is given by \eqref{Phix} and \eqref{Phicm}.
The positivity of the output covariance matrix when the input is a Dirac delta implies $\alpha\geq0$.
\end{proof}
\end{thm}
\subsection{Contractions}

\label{contr}

A contraction
by $\lambda=\frac{1}{\mu}$ behaves properly on the restricted subset $\mathfrak{G}^{(>)}_{\mu^2}$ of $\mathfrak{G}$ formed by the Gaussian states whose
covariance matrix admits symplectic eigenvalues larger than ${\mu}^2$.
Indeed all elements of $\mathfrak{G}^{(>)}_{\mu^2}$ will be mapped into
proper Gaussian output states by the contraction (and by linearity also the
convex hull of $\mathfrak{G}^{(>)}_{\mu^2}$ will be mapped into proper
output density operators). We will prove that any transformation with this property can be written as a contraction of $1/\mu$, followed by a transformation of the kind of Theorem \ref{gaussthmq}.
Let us first notice that:

\begin{lem}
A set $(K,\alpha)$ satisfies \eqref{positivity} for any $\sigma$ with
symplectic eigenvalues greater than $\mu^2$ iff $(\mu K,\;\alpha)$ satisfies
\eqref{positivity} for any $\sigma\geq\pm i\Delta$.
\begin{proof}
$\sigma$ has all the symplectic eigenvalues greater than $\mu^2$ iff $\sigma\geq\pm i\mu^2\Delta$, i.e. iff $\sigma'=\sigma/\mu^2$ is a state. Then \eqref{positivity} is satisfied for any $\sigma\geq\pm i\mu^2\Delta$ iff
\begin{equation}
\mu^2 K\sigma'K^T+\alpha\geq\pm i\Delta\qquad\forall\;\sigma'\geq\pm i\Delta\;,
\end{equation}
i.e. iff $(\mu K,\;\alpha)$ satisfies \eqref{positivity} for any $\sigma\geq\pm i\Delta$.
\end{proof}
\end{lem}

Then we can state the following result:

\begin{cor}
Any transformation associated with $(K,\alpha )$ satisfying
\eqref{positivity} for any state in $\mathfrak{G}_{\mu ^{2}}^{(>)}$ (i.e.
for any $\sigma \geq \pm i\mu ^{2}\Delta $) can be written as a contraction
of $1/\mu $, followed by a transformation satisfying \eqref{positivity} for
any state in $\mathfrak{G}$ (i.e. for any $\sigma \geq \pm i\Delta $).
\end{cor}

\begin{figure}[t]
\begin{center}
\includegraphics[trim=0pt 0pt 0pt 0pt, clip, width=\textwidth]{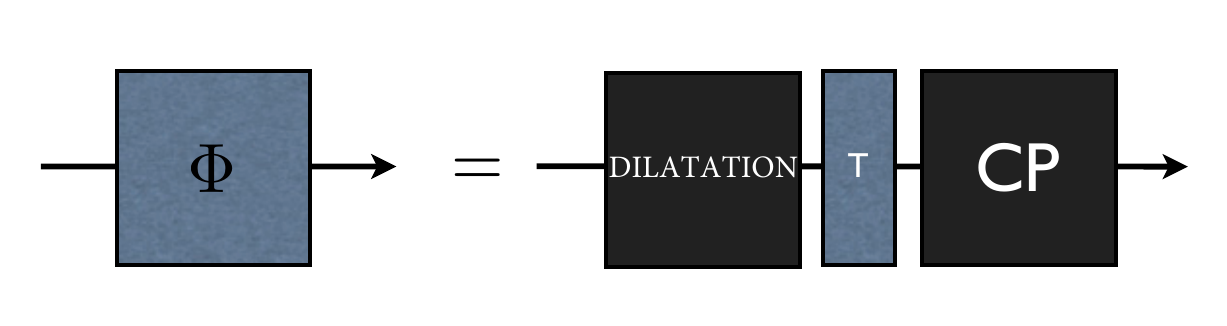}
\end{center}
\caption{(Color online) Pictorial representation of the decomposition of a
generic (not necessarily positive) Gaussian single-mode transformation $\Phi$
in terms of a dilatation, CP mapping and (possibly) a transposition. The
same decomposition applies also to the case of $n$-mode transformations when
no extra noise is added to the system, see Sec.~\protect\ref{S:MULTI}. }
\label{FigSys1}
\end{figure}

\section{One mode}

\label{S:ONE} Here we will give a complete classification of all the
one-mode maps (\ref{channel}) satisfying \eqref{positivity}.

We will need the following

\begin{lem}
A set $(K,\;\alpha)$ satisfies \eqref{positivity} iff
\begin{equation}  \label{det1}
\sqrt{\det\alpha}\geq1-|\det K|\;.
\end{equation}
\end{lem}

\begin{proof}
For one mode, $\sigma\geq0$ satisfies $\sigma\geq\pm i\Delta$ iff $\det\sigma\geq1$, and condition \eqref{positivity} can be rewritten as
\begin{equation}
\det \left( K\sigma K^{T}+\alpha \right) \geq 1,\qquad \forall \;\sigma \geq
0,\;\det \sigma \geq 1\;.  \label{posdet}
\end{equation}
To prove (\ref{posdet}) $\Longrightarrow $ (\ref{det1}) consider first the
case $\det K\neq 0.$ Choosing $\sigma $ such that $K\sigma K^{T}=\frac{|\det
K|}{\sqrt{\det \alpha }}\alpha ,$ we have $\sigma \geq 0,\det \sigma \geq 1.$
Inserting this into (\ref{posdet}), we obtain
\begin{equation*}
\left( 1+\frac{|\det K|}{\sqrt{\det \alpha }}\right) ^{2}\det \alpha \geq 1
\end{equation*}
or, taking square root,
\begin{equation*}
\left( 1+\frac{|\det K|}{\sqrt{\det \alpha }}\right) \sqrt{\det \alpha }\geq
1\;.
\end{equation*}
hence (\ref{det1}) follows.

If $\det K=0,$ then there is a unit vector $\mathbf{e}$ such that $K\mathbf{e}=0$. Choose $
\sigma =\epsilon ^{-1}\mathbf{e}\mathbf{e}^{T}+\epsilon \mathbf{e}_{1}\mathbf{e}_{1}^{T},$ where $\epsilon >0,$
and $\mathbf{e}_{1}$ is a unit vector orthogonal to $\mathbf{e}$. Then $\sigma \geq 0,\det
\sigma =1,$ and $K\sigma K^{T}=\epsilon A,$ where $A=K\mathbf{e}_{1}\mathbf{e}_{1}^{T}K^{T}
\geq 0$. Inserting this into (\ref{posdet}), we obtain
\begin{equation*}
\det \left( \epsilon A+\alpha \right) \geq 1,\quad \forall \;\epsilon \geq 0,
\end{equation*}
hence (\ref{det1}) follows.

To prove (\ref{det1}) $\Longrightarrow $  (\ref{posdet}), we use Minkowski's
determinant inequality
\begin{equation}
\sqrt{\det (A+B)}\geq \sqrt{\det A}+\sqrt{\det B}\qquad \forall \;A,\,B\geq
0\;.  \label{mink}
\end{equation}
We have for all $\sigma \geq 0,\det \sigma \geq 1,$
\begin{eqnarray}
&&\sqrt{\det \left( K\sigma K^{T}+\alpha \right) }\geq   \notag \\
&\geq &\left\vert \det K\right\vert \sqrt{\det \sigma }+\sqrt{\det \alpha }
\geq   \notag \\
&\geq &\left\vert \det K\right\vert +\sqrt{\det \alpha }\;\geq 1,
\label{ineq}
\end{eqnarray}
where in the last step we have used (\ref{det1}).
\end{proof}
To compare transformations satisfying \eqref{det1} with CP ones, we need also

\begin{lem}
A set $(K,\;\alpha)$ characterizes a completely positive transformation
(i.e. satisfies \eqref{CPTP}) iff
\begin{equation}
\sqrt{\det\alpha}\geq|1-\det K|\;.  \label{det2}
\end{equation}
\end{lem}

\begin{proof}
For one mode, using \eqref{det1mode},
\begin{equation}
\Delta_K=K\Delta K^T=\det K\;\Delta\;,
\end{equation}
and \eqref{CPTP} becomes
\begin{equation}\label{CPTPdet}
\alpha\geq\pm i(1-\det K)\Delta\;.
\end{equation}
Recalling \eqref{state1}, for linearity \eqref{CPTPdet} becomes exactly
\begin{equation}
\det\alpha\geq(1-\det K)^2\;.
\end{equation}
\end{proof}
We recall here that a complete classification of single mode CP
maps has been provided in Refs.~\cite{HOLEVO1,HOLEVO2}.

We are now ready to prove the main result of this section.

\begin{thm}
\label{thm:1mode} Any map $\Phi $ satisfying \eqref{positivity} can be
written as a dilatation possibly composed with the transposition, followed
by a completely positive map. In more detail, given a pair $(K,\;\alpha )$
satisfying \eqref{positivity},

\begin{description}
\item[a1] If
\begin{equation}
0\leq\det K\leq1\;,
\end{equation}
$\Phi$ is completely positive.

\item[a2] If
\begin{equation}
\det K>1\;,
\end{equation}
$\Phi$ can be written as a phase-space dilatation of parameter $\lambda=\sqrt{\det K}>1$, composed with
the symplectic transformation given by
\begin{equation}
S=\frac{K}{\sqrt{\det K}}\;,
\end{equation}
composed with the addition of Gaussian noise given by $\alpha$.

\item[b1] If
\begin{equation}
-1\leq\det K<0\;,
\end{equation}
$\Phi$ can be written as a transposition composed with a completely positive
map.

\item[b2] If
\begin{equation}
\det K<-1\;,
\end{equation}
$\Phi$ can be written as a dilatation of $\sqrt{|\det K|}$ composed with the
transposition, followed by the symplectic transformation given by
\begin{equation}
S=\frac{K}{\sqrt{|\det K|}}\;,
\end{equation}
composed with the addition of Gaussian noise given by $\alpha$.
\end{description}
\end{thm}

\begin{proof}
\begin{description}
\item[a] Let us start from the case
\begin{equation}
\det K\geq0\;.
\end{equation}
\begin{description}
\item[a1] If
\begin{equation}
0\leq\det K\leq1\;,
\end{equation}
\eqref{det1} and \eqref{det2} coincide, so $\Phi$ is completely positive.
\item[a2] If
\begin{equation}
\det K>1\;,
\end{equation}
we can write $K$ as
\begin{equation}
K=S\;\sqrt{\det K}\mathbbm{1}_2\;,
\end{equation}
where
\begin{equation}
S=\frac{K}{\sqrt{\det K}}
\end{equation}
is symplectic since $\det S=1$.
Then $\Phi$ can be written as a dilatation of $\sqrt{\det K}>1$, followed by the symplectic transformation given by $S$, composed with the addition of the Gaussian noise given by $\alpha$.
\end{description}
\item[b] If
\begin{equation}
\det K<0\;,
\end{equation}
we can write $K$ as
\begin{equation}
K=K'T\;,
\end{equation}
where $T$ is the one-mode transposition
\begin{equation}
T=\left(
    \begin{array}{cc}
      1 &  \\
       & -1 \\
    \end{array}
  \right)\;,
\end{equation}
and
\begin{equation}
\det K'=-\det K>0\;.
\end{equation}
From \eqref{det1} we can see that also $K'$ satisfies
\begin{equation}
\sqrt{\det\alpha}\geq1-\left|\det K'\right|\;,
\end{equation}
and we can exploit the classification with positive determinant, ending with the same decomposition with the addition of the transposition after (or before, since they commute) the eventual dilatation.
\end{description}
\end{proof}

\section{Multi-mode case}

\label{S:MULTI} In the multi-mode case, a classification as simple as the
one of theorem \ref{thm:1mode} does not exist. However, we will prove that
if $\Phi$ does not add any noise, i.e $\alpha=0$, the only solution to
\eqref{positivity} is a dilatation possibly composed with a (total)
transposition, followed by a symplectic transformation. We will also provide
examples that do not fall in any classification like \ref{thm:1mode}, i.e.
that are not composition of a dilatation, possibly followed by a (total)
transposition, and a completely positive map.

We will need the following lemma:

\begin{lem}
\label{lem:inf}
\begin{equation}
\inf_{\sigma\geq\pm i\Delta}\mathbf{w}^\dag\sigma\mathbf{w}=\left|\mathbf{w}
^\dag\Delta\mathbf{w}\right|\qquad\forall\;\;\mathbf{w}\in\mathbb{C}^{2n}\;.
\label{inf}
\end{equation}
\begin{proof}
~\paragraph{Lower bound}
The lower bound for the LHS is straightforward: for any $\sigma\geq\pm i\Delta$ and $\mathbf{w}\in\mathbb{C}^{2n}$ we have
\begin{equation}
\mathbf{w}^\dag\sigma\mathbf{w}\geq\pm i\mathbf{w}^\dag\Delta\mathbf{w}\;,
\end{equation}
and then
\begin{equation}
\inf_{\sigma\geq\pm i\Delta}\mathbf{w}^\dag\sigma\mathbf{w}\geq\left|\mathbf{w}^\dag\Delta\mathbf{w}\right|\;.
\end{equation}
\paragraph{Upper bound}
To prove the converse, let
\begin{equation*}
\mathbf{w}=\mathbf{w}_{1}+i\mathbf{w}_{2}\;,\qquad \mathbf{w}_{i}\in \mathbb{
R}^{2n}\;,
\end{equation*}
where without lost of generality we assume $\mathbf{w}_1\neq\mathbf{0}$.
Then
\begin{equation*}
\mathbf{w}^{\dag }\sigma \mathbf{w}=\mathbf{w}_{1}^{T}\sigma \mathbf{w}_{1}+\mathbf{w}_{2}^{T}\sigma
\mathbf{w}_{2},\quad \left\vert \mathbf{w}^{\dag }\Delta \mathbf{w}\right\vert
=2\left\vert \mathbf{w}_{1}^{T}\Delta \mathbf{w}_{2}\right\vert .
\end{equation*}
Assume first $\mathbf{w}_{1}^{T}\Delta \mathbf{w}_{2}\equiv \epsilon \neq0.$ Then we can
introduce the symplectic basis $\{\mathbf{e}_{j},\;\mathbf{h}_{j}\}_{j=1,\dots ,n}$, where
\begin{equation*}
\mathbf{e}_{1}=\frac{\mathbf{w}_{1}}{\sqrt{|\epsilon|}}\;,\qquad\mathbf{h}_{1}=\frac{\mathrm{sign}(\epsilon)\;\mathbf{w}_{2}}{\sqrt{|\epsilon|}}\;.
\end{equation*}
Expressed in this basis the question
\eqref{inf} reduces to the first mode, and the infimum is attained by the
matrix of the form
\begin{equation*}
\sigma =\left(
\begin{array}{cc}
1 & 0 \\
0 & 1
\end{array}
\right) \oplus \sigma _{n-1}\;,
\end{equation*}
 where $\sigma _{n-1}$ is any quantum correlation matrix in
the rest $n-1$ modes.

Consider next the case where $\mathbf{w}_{1}^{T}\Delta \mathbf{w}_{2}=0$ and $\mathbf{w}_{2}$ is not proportional to $\mathbf{w}_{1}$. In this context
we  introduce the symplectic basis $\{\mathbf{e}_{j},\;\mathbf{h}_{j}\}_{j=1,\dots ,n}$,
where
\begin{equation*}
\mathbf{e}_{1}=\mathbf{w}_{1}\;,\qquad \mathbf{e}_{2}=\mathbf{w}_{2}\;.
\end{equation*}
Accordingly the identity~\eqref{inf} reduces to the
first two modes, and the infimum is attained by the matrices of the form
\begin{equation*}
\sigma (\epsilon )=\left(
\begin{array}{cc}
\epsilon  & 0 \\
0 & \epsilon ^{-1}
\end{array}
\right) \oplus \left(
\begin{array}{cc}
\epsilon  & 0 \\
0 & \epsilon ^{-1}
\end{array}
\right) \oplus \sigma _{n-2}\;,
\end{equation*}
where $\sigma _{n-2}$ is any quantum correlation matrix in the rest $n-2$
modes, and $\epsilon \rightarrow 0$.

{Finally, if $\mathbf{w}_{2}=c\;\mathbf{w}_1,\;c\in\mathbb{R}$, we  introduce the symplectic basis $\{\mathbf{e}_{j},\;\mathbf{h}_{j}\}_{j=1,\dots ,n}$, where $\mathbf{e}_{1}=\mathbf{w}_{1}$. The question \eqref{inf} reduces to the
first mode, and the infimum is attained by the matrices of the form
\begin{equation*}
\sigma (\epsilon )=\left(
\begin{array}{cc}
\epsilon  & 0 \\
0 & \epsilon ^{-1}
\end{array}
\right) \oplus \sigma _{n-1}\;,
\end{equation*}
where $\sigma _{n-1}$ is any quantum correlation matrix in the rest $n-1$
modes, and $\epsilon \rightarrow 0$.}
\end{proof}
\end{lem}

A simple consequence of lemma \ref{lem:inf} is

\begin{lem}
Any $\alpha$ satisfying \eqref{positivity} for some $K$ is positive
semidefinite.
\end{lem}

\begin{proof}
The constraint \eqref{positivity} implies
\begin{equation}
\left(K^T\mathbf{k}\right)^T\sigma\left(K^T\mathbf{k}\right)+\mathbf{k}^T\alpha\mathbf{k}\geq0\label{kpos}
\end{equation}
for any $\sigma\geq\pm i\Delta$ and $\mathbf{k}\in\mathbb{R}^{2n}$.
Taking the inf over $\sigma\geq\pm i\Delta$, and exploiting lemma \ref{lem:inf} with $\mathbf{w}=K^T\mathbf{k}$, we get
\begin{equation}\label{kak}
\mathbf{k}^T\alpha\mathbf{k}\geq0\;,
\end{equation}
i.e. $\alpha$ is positive semidefinite.
In deriving \eqref{kak} we have used that, since $\Delta$ is antisymmetric, $\mathbf{k}\Delta\mathbf{k}^T=0$ for any real $\mathbf{k}$.
\end{proof}

The Lemma \ref{lem:inf} allows us to rephrase the problem: indeed, the constraint
\eqref{positivity} can be written as
\begin{equation}
(K^T\mathbf{w})^\dag\sigma(K^T\mathbf{w})+\mathbf{w}^\dag\alpha\mathbf{w}
\geq\left|\mathbf{w}^\dag\Delta\mathbf{w}\right|\;,
\end{equation}
$\forall\;\sigma\geq\pm i\Delta$, $\forall \mathbf{w}\in\mathbb{C}^{2n}$.
Taking the inf over $\sigma$ in the LHS we hence get
\begin{equation}  \label{posw}
\left|\mathbf{w}^\dag\Delta_K\mathbf{w}\right|+\mathbf{w}^\dag\alpha\mathbf{w
}\geq\left|\mathbf{w}^\dag\Delta\mathbf{w}\right|\;,\quad\forall\;\mathbf{w}
\in\mathbb{C}^{2n}\;,
\end{equation}
with $\Delta_K$ as in Eq.~(\ref{DELTAK}). Notice that, as for the complete
positivity constraint \eqref{CPTP}, since $K$ enters in \eqref{posw} only
through $\left|\mathbf{w}^\dag\Delta_K\mathbf{w}\right|$, whether given $K$
and $\alpha$ satisfy \eqref{positivity} depends not on the entire $K$ but
only on $\Delta_K$.

The easiest way to give a general classification of the channels satisfying
\eqref{posw} (and then \eqref{positivity}) would seem choosing a basis in
which $\Delta$ is as in \eqref{deltacan}, and then try to put the
antisymmetric matrix $\Delta_K$ in some canonical form using symplectic
transformations preserving $\Delta$. However, the complete classification of
antisymmetric matrices under symplectic transformations is very involved
\cite{LR}, and in the multi-mode case the problem simplifies only if we
consider maps $\Phi$ that do not add noise, since in this case the
constraint \eqref{posw} rules out almost all the equivalence classes. In the
general case, we will provide examples showing the other possibilities.

\subsection{No noise}

The main result of this section is the classification of the maps $\Phi$
that do not add noise ($\alpha=0$) and satisfy \eqref{positivity}:

\begin{thm}
\label{thm:nonoise} A map $\Phi$ with $\alpha=0$ satisfying
\eqref{positivity} can always be decomposed as a dilatation~(\ref{DILATATION}
), possibly composed with the transposition, followed by a
symplectic $S$ transformation: i.e.
\begin{equation}
K=S\;\kappa\mathbbm{1}_{2n}\qquad\text{or}\qquad K=S\;T\;\kappa\mathbbm{1}
_{2n}\;,
\end{equation}
with $\kappa\geq1$.
\end{thm}

\begin{proof}
With $\alpha=0$ and
\begin{equation}
\mathbf{w}=\mathbf{w}_1+i\mathbf{w}_2\;,\qquad\mathbf{w}_i\in\mathbb{R}^{2n}\;,
\end{equation}
\eqref{posw} becomes
\begin{equation}
\left|\mathbf{w}_1^T\Delta_K\mathbf{w}_2\right|\geq\left|\mathbf{w}_1^T\Delta\mathbf{w}_2\right|\;,\label{Deltamaj}
\end{equation}
i.e. all the matrix elements of $\Delta_K$ are in modulus bigger than the corresponding ones of $\Delta$ in \emph{any} basis. In particular, if some matrix element $\Delta_K^{ij}$ vanishes, also $\Delta^{ij}$ must vanish.
Let us choose a basis in which $\Delta_K$ has the canonical form
\begin{equation}
\Delta_K=\bigoplus_{i=1}^{\frac{r}{2}}\left(
                                         \begin{array}{cc}
                                            & 1 \\
                                           -1 &  \\
                                         \end{array}
                                       \right)\oplus0_{2n-r}\;,
\end{equation}
where
\begin{equation}
r\equiv\mathrm{rank}\,\Delta_K\;.
\end{equation}
For \eqref{Deltamaj}, in this basis $\Delta$ must be of the form
\begin{equation}
\Delta=\bigoplus_{i=1}^{\frac{r}{2}}\left(
                                         \begin{array}{cc}
                                            &  \lambda_i\\
                                           -\lambda_i &  \\
                                         \end{array}
                                       \right)\oplus0_{2n-r}\;,\qquad|\lambda_i|\leq1\;.
\end{equation}
Since $\Delta$ has full rank, there cannot be zeroes in its decomposition, so $r$ must be $2n$.

We will prove that all the eigenvalues $\lambda_i$ must be equal. Let us take two eigenvalues $\lambda$ and $\mu$, and consider the restriction of $\Delta$ and $\Delta_K$ to the subspace associated to them:
\begin{equation}
\Delta_K=\left(
           \begin{array}{cc|cc}
              & 1 &  &  \\
             -1 &  &  & \\
             \hline
              &  &  & 1 \\
              &  & -1 &  \\
           \end{array}
         \right)\qquad\Delta=\left(
           \begin{array}{cc|cc}
              & \lambda &  &  \\
             -\lambda &  &  & \\
             \hline
              &  &  & \mu \\
              &  & -\mu &  \\
           \end{array}
         \right)\;.
\end{equation}
If we change basis with the rotation matrix
\begin{eqnarray}
&R=\left(
    \begin{array}{cc}
      \cos\theta\;\mathbbm{1}_2 & -\sin\theta\;\mathbbm{1}_2 \\
      \sin\theta\;\mathbbm{1}_2 & \cos\theta\;\mathbbm{1}_2 \\
    \end{array}
  \right)\;,&\nonumber \\ &\Delta\mapsto R\Delta R^T\;,\qquad  \Delta_K\mapsto R\Delta_K R^T\;,&
\end{eqnarray}
$\Delta_K$ remains of the same form, while $\Delta$ acquires off-diagonal elements proportional to $\lambda-\mu$. Since for \eqref{Deltamaj} the off-diagonal elements of $\Delta$ must vanish also in the new basis, the only possibility is $\lambda=\mu$. Then all the $\lambda_i$ must be equal, and $\Delta_K$ must then be proportional to $\Delta$:
\begin{equation}
\Delta_K=\frac{1}{\lambda}\Delta\;,\qquad0<|\lambda|\leq1\;,\label{deltaprop}
\end{equation}
where we have put all the $\lambda_i$ equal to $\lambda\neq0$ (since $\Delta$ is nonsingular they cannot vanish).
Relation \eqref{deltaprop} means
\begin{equation}
K\Delta K^T=\frac{1}{\lambda}\Delta\;,
\end{equation}
i.e.
\begin{equation}
\left(\sqrt{|\lambda|}\;K\right)\;\Delta\;\left(\sqrt{|\lambda|}\;K\right)^T=\mathrm{sign}(\lambda)\Delta\;.
\end{equation}
If $0<\lambda\leq1$, we can write $K$ as a dilatation of
\begin{equation}
\kappa=\frac{1}{\sqrt{\lambda}}\;,
\end{equation}
composed with a symplectic transformation given by
\begin{equation}
S=\sqrt{\lambda}\;K\;,
\end{equation}
i.e.
\begin{equation}
K=S\;\kappa\mathbbm{1}_{2n}\;,\qquad S\Delta S^T=\Delta\;.
\end{equation}
If $-1\leq\lambda<0$, since the total transposition $T$ changes the sign of $\Delta$:
\begin{equation}
T\Delta T^T=-\Delta\;,
\end{equation}
we can write $K$ as a dilatation of
\begin{equation}
\kappa=\frac{1}{\sqrt{|\lambda|}}\;,
\end{equation}
composed with $T$ followed by a symplectic transformation:
\begin{equation}
K=S\;T\;\kappa\mathbbm{1}_{2n}\;,\qquad S\Delta S^T=\Delta\;.
\end{equation}
\end{proof}

\subsection{Examples with nontrivial decomposition}

If $\alpha\neq0$, a decomposition as simple as the one of theorem \ref
{thm:nonoise} does no more exist: here we will provide some examples in
which the canonical form of $\Delta_K$ is less trivial, and that do not fall
in any classification like the precedent one. Essentially, they are all
based on this observation:

\begin{prop}
If $\alpha$ is the covariance matrix of a quantum state, i.e. $\alpha\geq\pm
i\Delta$, the constraint \eqref{positivity} is satisfied by any $K$.
\end{prop}

Since for one mode the decomposition of theorem \ref{thm:1mode} holds, we
will provide examples with two-mode systems.

We will always consider bases in which
\begin{equation}
\Delta=\left(
\begin{array}{cc|cc}
& 1 &  &  \\
-1 &  &  &  \\ \hline
&  &  & 1 \\
&  & -1 &
\end{array}
\right)\;.
\end{equation}

\subsubsection{Partial transpose}

The first example is the partial transpose of the second subsystem, composed
with a dilatation of $\sqrt{\nu}$ and the addition of the covariance matrix
of the vacuum as noise:
\begin{equation}
K=\sqrt{\nu}\left(
\begin{array}{cc}
\mathbbm{1}_2 &  \\
& T_2
\end{array}
\right)\;,\qquad\nu>0\;,\qquad\alpha=\mathbbm{1}_4\;.
\end{equation}
In this case we have
\begin{equation}  \label{deltaKt}
\Delta_K=\left(
\begin{array}{cc|cc}
& \nu &  &  \\
-\nu &  &  &  \\ \hline
&  &  & -\nu \\
&  & \nu &
\end{array}
\right)\;,
\end{equation}
and $i(\Delta-\Delta_K)$ has eigenvalues
\begin{equation}
\pm(1+\nu)\qquad\pm(1-\nu)\;,
\end{equation}
so that one of them is $\left|1+|\nu|\right|>1$, and the complete positivity
requirement \eqref{CPTP}
\begin{equation}
\mathbbm{1}_4\geq\pm i(\Delta-\Delta_K)
\end{equation}
cannot be fulfilled by any $\nu\neq0$.

We will prove that this map cannot be written as a dilatation, possibly
composed with the transposition, followed by a completely positive map.
Indeed, suppose we can write $K$ as
\begin{equation}
K=K^{\prime}\;\lambda\mathbbm{1}_4\qquad\text{or}\qquad
K=K^{\prime}\;T_4\;\lambda\mathbbm{1}_4\;,\qquad\lambda\geq1\;.
\end{equation}
Then
\begin{equation}
\Delta_{K^{\prime}}=\pm\frac{1}{\lambda^2}\Delta_K
\end{equation}
is always of the form \eqref{deltaKt} with
\begin{equation}
\nu^{\prime}=\pm\frac{\nu}{\lambda^2}\;,
\end{equation}
and also the transformation with $K^{\prime}$ cannot be completely positive.

\subsubsection{\emph{Q} exchange}

As second example, we take for the added noise $\alpha$ still the covariance
matrix of the vacuum, and for the matrix $K$ the partial transposition of
the first mode composed with the exchange of $Q^1$ and $Q^2$ followed by a
dilatation of $\sqrt{\nu}$:
\begin{equation}
\alpha=\mathbbm{1}_4\geq\pm i\Delta\;,\quad K=\sqrt{\nu}\left(
\begin{array}{cc|cc}
&  & 1 &  \\
& -1 &  &  \\ \hline
1 &  &  &  \\
&  &  & 1
\end{array}
\right)\;,\quad\nu>0\;.
\end{equation}
With this choice,
\begin{equation}  \label{deltakex}
\Delta_K=\left(
\begin{array}{cc|cc}
&  &  & \nu \\
&  & \nu &  \\ \hline
& -\nu &  &  \\
-\nu &  &  &
\end{array}
\right)\;.
\end{equation}
The transformation is completely positive iff
\begin{equation}  \label{CPTPex}
\mathbbm{1}_4\geq\pm i(\Delta-\Delta_K)\;,
\end{equation}
and since the eigenvalues of $i(\Delta-\Delta_K)$ are
\begin{equation}
\pm\sqrt{1+\nu^2}\;,
\end{equation}
\eqref{CPTPex} is never fulfilled for any $\nu\neq0$.

As before, we will prove that this map cannot be written as a dilatation,
possibly composed with the transposition, followed by a completely positive
map. Indeed, suppose we can write $K$ as
\begin{equation}
K=K^{\prime}\;\lambda\mathbbm{1}_4\qquad\text{or}\qquad
K=K^{\prime}\;T_4\;\lambda\mathbbm{1}_4\;,\qquad\lambda\geq1\;.
\end{equation}
Then
\begin{equation}
\Delta_{K^{\prime}}=\pm\frac{1}{\lambda^2}\Delta_K
\end{equation}
is always of the form \eqref{deltakex} with
\begin{equation}
\nu^{\prime}=\pm\frac{\nu}{\lambda^2}\;,
\end{equation}
and also the transformation with $K^{\prime}$ cannot be completely positive.

\section{Conclusions}

\label{S:CON}

In this paper we have explored both at the classical and quantum level the
set of linear transformations sending the set
of Gaussian states into itself without imposing any further requirement,
such as positivity. We have proved that the action on the covariance matrix
and on the first moment must be linear, and we have found the form of the
action on the characteristic function. Focusing on the quantum case, for one
mode we have obtained a complete classification, stating that the only not
CP transformations in the set are actually the total transposition and the
dilatations (and their compositions with CP maps). The same result holds
also in the multi-mode scenario, but it needs the further hypothesis of
homogeneous action on the covariance matrix, since we have shown the
existence of non-homogeneous transformations belonging to the set but not
falling into our classification.

Despite the set $\mathfrak{F}$ of quantum states that are sent into positive operators by any dilatation is known to strictly contain the convex hull of Gaussian states $\mathfrak{C}$ even in the one-mode case\cite{WERNER}, the dilatations are then confirmed to be (at least in the single mode or in
the homogeneous action cases) the only transformation in the class
\eqref{channel} that can act as a probe for $\mathfrak{C}$.

\section{Acknowledgements}

GdP thanks F. Poloni for useful discussions. This work was supported in part
by the ERC through the Advanced Grant n. 321122 SouLMan.

\appendix
\section{Unboundedness of dilatations}\label{app}
\begin{thm}\label{unbounded}
For any $\lambda\neq\pm1$ the phase-space dilatation by $\lambda$ is not bounded in the Banach space $\mathfrak{T}$ of trace-class operators.
\begin{proof}
Fix $\lambda\neq\pm1$, and let $\Theta$ be the phase-space dilatation by $\lambda$.
Suppose $\Theta$ to be bounded, i.e.
\begin{equation}\label{bounded}
\left\|\Theta\left(\hat{X}\right)\right\|_1\leq\left\|\Theta\right\|\;\left\|\hat{X}\right\|_1\qquad\forall\;\hat{X}\in\mathfrak{T}\;.
\end{equation}
Let also
\begin{equation}
p_n^{(m)}:=\langle n|\Theta\left(|m\rangle\langle m|\right)|n\rangle\;.
\end{equation}
Eq. \eqref{bounded} implies
\begin{equation}
\sum_{n=0}^\infty\left|p_n^{(m)}\right|\leq\|\Theta\|\qquad\forall\;m\in\mathbb{N}\;.
\end{equation}
The moment generating function of $p^{(m)}$ is \cite{WERNER}
\begin{equation}\label{mgf}
g_m(q):=\sum_{n=0}^\infty p_n^{(m)}\;e^{-i\,n\,q}=\frac{1-\tau}{1-\tau\,e^{-i\,q}}\left(\frac{1-\tau\,e^{i\,q}}{e^{i\,q}-\tau}\right)^m\;,
\end{equation}
where $q\in\mathbb{R}$ and
\begin{equation}\label{tau2}
\tau:=\frac{\lambda^2-1}{\lambda^2+1}\;.
\end{equation}
Define
\begin{equation}\label{am}
a_m:=\frac{1-\tau}{\sqrt[3]{m\,\tau(1+\tau)}}\;.
\end{equation}
Let $\phi\in C_c^\infty(\mathbb{R})$ be an infinitely differentiable test function with compact support.
We must then have
\begin{equation}\label{cond}
\sum_{n=0}^\infty\phi\left(a_m\left(n-\lambda^2m\right)\right)\;p_n^{(m)}\leq\|\phi\|_\infty\;\|\Theta\|\;.
\end{equation}
Expressed in terms of the Fourier transform of $\phi$
\begin{equation}
\widetilde{\phi}(k)=\int_{-\infty}^\infty\phi(x)\;e^{i\,k\,x}\;dx\;,
\end{equation}
\eqref{cond} becomes
\begin{eqnarray}
\sum_{n=0}^\infty\left(\int_{-\infty}^\infty\widetilde{\phi}(k)\;e^{i\,\lambda^2\,m\,a_m\,k}\;e^{-i\,k\,a_m\,n}\;\frac{dk}{2\pi}\right)p_n^{(m)}\leq\nonumber\\
\leq\|\phi\|_\infty\;\|\Theta\|\;.\quad
\end{eqnarray}
Since the sum of the integrands is dominated by the integrable function
\begin{equation*}
\frac{\|\Theta\|}{2\pi}\left|\widetilde{\phi}(k)\right|\;,
\end{equation*}
we can bring the sum inside the integral, getting
\begin{equation}
\int_{-\infty}^\infty\widetilde{\phi}(k)\;g_m\left(a_mk\right)\;e^{i\,\lambda^2\,m\,a_m\,k}\;\frac{dk}{2\pi}\leq\|\phi\|_\infty\;\|\Theta\|\;.
\end{equation}
Since for any $k$
\begin{equation}\label{limitm}
\lim_{m\to\infty}\left(g_m\left(a_m\,k\right)\;e^{i\,\lambda^2\,m\,a_m\,k}\right)=e^{\frac{i\,k^3}{3}}
\end{equation}
(see subsection \ref{note}), by the dominated convergence theorem
\begin{eqnarray}
\lim_{m\to\infty}\int_{-\infty}^\infty\widetilde{\phi}(k)\;g_m\left(a_mk\right)\;e^{i\,\lambda^2\,m\,a_m\,k}\;\frac{dk}{2\pi} =\nonumber\\
=\int_{-\infty}^\infty\widetilde{\phi}(k)\;e^\frac{i\,k^3}{3}\;\frac{dk}{2\pi}=\int_{-\infty}^\infty\phi(x)\;\mathrm{Ai}(x)\;dx\;,
\end{eqnarray}
where $\mathrm{Ai}(x)$ is the Airy function.
Now we get
\begin{equation}\label{limit}
\int_{-\infty}^\infty \mathrm{Ai}(x)\;\phi(x)\;dx\leq \|\Theta\|\;\|\phi\|_\infty\qquad\forall\;\phi\in C_c^\infty(\mathbb{R})\;.
\end{equation}
Since the Airy function is continuous and the set of its zeroes has no accumulation points (except $-\infty$), there exists a sequence of test functions $\phi_r\in C_c^\infty(\mathbb{R})$, $r\in\mathbb{N}$ with $\|\phi_r\|_\infty=1$ approximating $\mathrm{sign}\left(\mathrm{Ai}(x)\right)$, i.e. such that
\begin{equation}
\lim_{r\to\infty}\int_{-\infty}^\infty \mathrm{Ai}(x)\;\phi_r(x)\;dx=\int_{-\infty}^\infty\left|\mathrm{Ai}(x)\right|dx=\infty\;,
\end{equation}
implying
\begin{equation}
\|\Theta\|=\infty\;.
\end{equation}
\end{proof}
\end{thm}
\subsection{Computation of the limit in (A11)}\label{note}
Here we compute explicitly the limit in \eqref{limitm}.
It is better to rephrase it in terms of
\begin{equation}
q:=a_m\,k\;,\qquad q\to0
\end{equation}
(remember that $a_m\sim1/\sqrt[3]{m}$).
Putting together \eqref{limitm}, \eqref{mgf}, \eqref{tau2} and \eqref{am}, we have to compute
\begin{equation} \label{ffd}
\lim_{q\to0}\left(\frac{1-\tau}{1-\tau\,e^{-i\,q}}\left(\frac{1-\tau\,e^{i\,q}}{e^{i\,q}-\tau}\;e^{i\,\frac{1+\tau}{1-\tau}\,q}\right)^\frac{k^3(1-\tau)^3}{q^3\,\tau(1+\tau)}\right)\overset{?}{=}e^\frac{i\,k^3}{3}\;.
\end{equation}
The first term on the left-hand-side tends to one. The second term on the left-hand-side instead can be treated via Taylor expansion, i.e.
\begin{equation}
\frac{1-\tau\,e^{i\,q}}{e^{i\,q}-\tau}\;e^{i\,\frac{1+\tau}{1-\tau}\,q}=1+\frac{i\,q^3\,\tau(1+\tau)}{3(1-\tau)^3}+\mathcal{O}\left(q^5\right)
\end{equation}
for $q\to0$.
This gives
\begin{eqnarray}
&&\lim_{q\to0}\left(\frac{1-\tau\,e^{i\,q}}{e^{i\,q}-\tau}\;e^{i\,\frac{1+\tau}{1-\tau}\,q}\right)^\frac{k^3(1-\tau)^3}{q^3\,\tau(1+\tau)}=\nonumber\\
&&=\lim_{q\to0}\left(1+\frac{i\,q^3\,\tau(1+\tau)}{3(1-\tau)^3}+\mathcal{O}\left(q^5\right)\right)^\frac{k^3(1-\tau)^3}{q^3\,\tau(1+\tau)}=\nonumber\\
&&=e^\frac{i\,k^3}{3}\;,
\end{eqnarray}
which proves the identity of (\ref{ffd}).

\end{document}